\documentclass[12pt]{article}

\usepackage{graphicx}
\usepackage{hyperref,url}
\usepackage{color}
\usepackage{amssymb,latexsym}
\usepackage{amsmath,amsthm}
\usepackage{enumerate}
\usepackage{tikz}
\usetikzlibrary{shapes,positioning}
\usepackage{latexsym,amscd,amssymb,verbatim,hyperref}
\usepackage{amsfonts}
\usepackage{dsfont}
\usepackage{stmaryrd}
\usepackage{tikz}
\usepackage{caption}
\usepackage{subcaption}
\usepackage[affil-it]{authblk} 

\usepackage{palatino}

\usepackage[applemac]{inputenc}
\usepackage[english]{babel}

\usepackage{geometry}
\usepackage{setspace} 

\newtheorem{theoA}{\bf Theorem}[section]

\newtheorem{definition-theorem}[theoA]{Definition-Theorem}

\newcommand{\net}{\mathfrak N} 
\newcommand{\op}[1]{\mathsf{Op}(#1)} 
\newcommand{\Uni}{\text{U}} 
\newcommand{\Uf}{\mathsf U} 
\renewcommand{\dim}{\text{Dim}} 

\newcommand{\suchthat}{\colon} 
\newcommand{\bbR}{\mathbb R} 
\newcommand{\bbC}{\mathbb C}
\newcommand{\bol}{\boldsymbol} 
\newcommand{\vect}[1]{\boldsymbol{#1}} 
\newcommand{\deq}{\stackrel{\text{df}}{=}} 
\newcommand{\disand}{\qquad \text{and}\qquad} 

\newcommand{\one}{\mathds{1}}
\newcommand{\calO}{\mathcal{O}}

\newcommand{\tr}{\mathrm{tr}_}
\newcommand{\hil}{\mathcal{H}^}
\newcommand{\bra}[1]{\langle #1 |}
\newcommand{\ket}[1]{| #1 \rangle}
\newcommand{\ketbra}[2]{| #1 \rangle \langle #2 |}
\newcommand{\braket}[2]{\langle #1 | #2 \rangle}
\newcommand{\comp}[1]{\overline{#1}}
\newcommand{\cnot}{\text{CNOT}}

\def \be {\begin{equation}}
\def \ee {\end{equation}}
\def \bes {\begin{equation*}}
\def \ees {\end{equation*}}
\def \baa {\begin{align}}
\def \eaa {\end{align}}
\def \baas {\begin{align*}}
\def \eaas {\end{align*}}
\def \bea {\begin{eqnarray}}
\def \eea {\end{eqnarray}}
\def \beas {\begin{eqnarray*}}
\def \eeas {\end{eqnarray*}}

\begin{document}

\title{The Cost of Quantum Locality}

\author{
Charles Alexandre B\'edard}
\affil{\small {Universit\`a della Svizzera italiana}\\
\footnotesize \emph{charles.alexandre.bedard@usi.ch}}

\maketitle
\begin{abstract}
It has been more than 20 years since Deutsch and Hayden demonstrated that quantum systems can be completely described locally --- notwithstanding Bell's theorem.
More recently, Raymond-Robichaud proposed two other approaches to the same conclusion. 
In this paper, all these means of describing quantum systems locally are proved formally equivalent.
The cost of such descriptions is then quantified by the dimensionality of their underlining space. The number of degrees of freedom of a single qubit's local description is shown to grow exponentially with the total number of qubits considered as a global system. 
%
This apparently unreasonable cost to describe such a small system in a large Universe is nonetheless shown to be expected.
Finally, structures that supplement the universal wave function are investigated.
\end{abstract}


\section{Introduction}\label{secintro}

It is still a widespread belief that a complete description of a composite entangled quantum system cannot be obtained by descriptions of the parts, if those are expressed independently of what happens to other parts. This apparently holistic feature of entangled quantum states 
entails 
violation of Bell inequalities~\cite{bell1964, aspect1982experimental} and quantum teleportation~\cite{bennett1993teleporting}, which are repeatedly invoked to sanctify the ``non-local'' character of
 quantum theory. But this widespread belief has been proven false more than twenty years ago by Deutsch and Hayden~\cite{deutsch2000information}, who by the same token, provided an entirely local explanation of Bell-inequality violations and teleportation.

Descriptions of dynamically isolated --- but possibly entangled --- systems~$A$ and~$B$ 
are \emph{local}\footnote{
After Bell, it has become conventional wisdom to equate locality with a possible explanation by a local hidden variable theory. However, local hidden variables are only one way in which locality can be instantiated~\cite{brassard2019parallel}. Here, locality is taken in the spirit of Einstein: ``the real factual situation of the system~$S_2$ is independent of what is done with the system~$S_1$, which is spatially separated from the former'' \cite{schilppalbert1970}.   }
 if that of~$A$ is unaffected by any process system~$B$ may undergo, and \textit{vice versa}. 
  The descriptions are \emph{complete} if they can predict the distributions of any measurement performed on the whole system~$AB$. For instance, if $AB$ is in a pure entangled state~$\ket{\Psi}^{AB}$, the reduced density matrices
\bes
\rho^A= \tr{B} \ketbra{\Psi}{\Psi} \disand \rho^B = \tr{A} \ketbra{\Psi}{\Psi}
\ees
are local but incomplete descriptions. If instead the descriptions of~$A$ and~$B$ are both taken to be the global wave function~$\ket{\Psi}^{AB}$, then one finds a complete but non-local account.

Following Gottesman's~\cite{gottesman1999heisenberg} quantum computation in the Heisenberg picture, Deutsch and Hayden define so-called \emph{descriptors} for individual qubits. 
 Such a mode of description is shown to be both local and complete, hence vindicating the locality of quantum theory. 
 More recently, Raymond-Robichaud reached the same conclusion by two different routes. In a generic approach, he showed~\cite{raymond2018equivalence, brassard2017equivalence} that any non-signalling theory with reversible operations can be reformulated in terms of so-called noumenal states, which describe the systems in a local and complete way. 
 Being a special case of such non-signalling theories, quantum theory then inherits what I shall refer to as \emph{quantum noumenal states}. Raymond-Robichaud's other approach~\cite{raymond2018equivalence, raymond2020local} 
addresses the question of locality directly within quantum theory by expressing local descriptions of quantum systems in terms of yet another object, \emph{evolution matrices}. In the next section, 
the formalisms of descriptors, noumenal states and evolution matrices are overviewed. 
%
\begin{center}
\begin{tabular}{|l|c|c|}
\hline 
\centering{\vphantom{$\int_A^B$}Mode of description of quantum systems} & Local & Complete \\
\hline \hline
\vphantom{$\int_A^B$}Reduced density matrices~$\rho^A$ and~$\rho^B$  & Yes & No \\
\hline
\vphantom{$\int_A^B$}Global wave function~$\ket{\Psi}^{AB}$ &No& Yes \\
\hline
\vphantom{$\int_A^B$}\small{Descriptors, quantum noumenal states \& evolution matrices} &Yes& Yes \\
\hline
\end{tabular}
\end{center}

In~\S\ref{secequiv}, a key result is proven: The three aforementioned approaches to quantum locality are formally equivalent. 
This equivalence uncovers an important drawback related to the amount of information encompassed in local descriptions. As investigated in~\S\ref{seccost}, the dimensionality of the state space of a system as tiny as a qubit is shown to scale exponentially with the size of the whole system considered. As it shall be argued, however, an exponential cost in terms of number of degrees of freedom was to be expected of any local and complete description of quantum systems. 
Finally, the local descriptions are shown to feature extra structures that supplement the universal wave function, as explored in~\S\ref{secwhatmore}.

%
%

\section{Preliminaries} \label{secpre}

The formalisms of descriptors, noumenal states and evolution matrices are briefly covered in this section. For an elementary and detailed introduction to the Deutsch--Hayden descriptors, see~\cite{bedard2020abc}.


\subsection{Descriptors}
\label{secdes}

Motivated by the quantum computation in the Heisenberg picture, the Deutsch--Hayden descriptors~\cite{deutsch2000information} encode the information of individual qubits in a pair of time-evolved observables. 
Let $\net$ be a computational network of~$n$ qubits initialized at time~$0$ in the Schr\"odinger state $\ket 0^{\otimes n}$, more conveniently denoted~$\ket 0$, which thereby determines the Heisenberg reference vector\footnote{Calling it a state is a misnomer, since its fixed value generally contains no information whatsoever about the qubits after they evolve. It is the time-evolved pair of observables, the descriptor, that carries the information.}. 
The apparent lack of generality of the network~$\net$ (including its reference vector fixed to~$\ket 0$) is lifted by its ability to simulate any other quantum system with arbitrary accuracy~\cite{deutsch1989quantum} (including a network with a different reference vector). 
At time~$0$, the \emph{descriptor} of qubit~$i$ is given by
\bes
\bol q_i(0) = \one^{i-1} \otimes (\sigma_x, \sigma_z) \otimes \one^{n-i} \,,
\ees
where~$\sigma_x$ and~$\sigma_z$ are the corresponding Pauli matrices. The descriptor is therefore a vector of two\footnote{Deutsch and Hayden originally defined the descriptor with a third component,~$\sigma_y$. It is however redundant.} components, each of which being an operator acting on the whole network. 
Suppose that  between the discrete times $t-1$ and $t$, only one gate is performed, whose matrix representation on the whole network is denoted~$G_t$. Let $U= G_t V$, where $V$ consists of all gates from time $0$ to $t-1$. The descriptor of qubit~$i$ at time~$t$ is given by
\bes
\bol q_i(t) = U^\dagger \bol q_i(0) U \,.
\ees
  Alternatively, $\bol q_i(t)$ can be expressed as
\bes
\bol q_i(t) = \Uf^\dagger_{G_t}(\bol q(t-1)) \bol q_i(t-1) \Uf_{G_t}(\bol q(t-1)) \,,
\ees
where~$\bol q(\cdot)=(\bol q_1(\cdot), \dots, \bol q_n(\cdot))$ is the~$2n$-component object that encodes the descriptor of each qubit at the corresponding time and $\Uf_{G_t}(\cdot)$ is a fixed operator valued function of some components of its argument. The function satisfies the defining equation $\Uf_{G_t}(\bol q(0))=G_t$, which is guaranteed to exist by the ability for the components of $\bol q(0)$ to multiplicatively generate a basis of the operators acting on $\net$. 
In other words, any linear operator~$G_t$ can be expressed as a polynomial in the~$2n$ matrices~$q_{1x}(0), q_{1z}(0), \dots, q_{nx}(0), q_{nz}(0)$, and~$\Uf_{G_t}(\bol q(0))$ is one such polynomial.
The two displayed expressions for $\bol q_i(t)$ are indeed equivalent,
\beas
 V^\dagger G_t^\dagger  \bol q_i(0) G_t V 
&=& V^\dagger \Uf^\dagger_{G_t}(\bol q(0)) V V^\dagger \bol q_i(0)V V^\dagger \Uf_{G_t}( \bol q(0) ) V \\
&=& \Uf^\dagger_{G_t}(V^\dagger \bol q(0) V) V^\dagger \bol q_i(0) V \Uf_{G_t}(V^\dagger \bol q(0) V)\\
&=& \Uf^\dagger_{G_t}(\bol q(t-1)) \bol q_i(t-1) \Uf_{G_t}(\bol q(t-1)) \,.
\eeas
The locality of the descriptors is recognized by the following. If the gate~$G_t$ acts only on qubits of the subset $I\subset \{1,2,\ldots,n\}$, then its functional representation~$\Uf_{G_t}$ shall only depend on components of~$\bol q_k(t-1)$, for~$k \in I$. So for any~$j \notin I$, the 
descriptor~$\bol q_{j}(t-1)$ commutes with~$\Uf_{G_t}(\bol q(t-1))$, and thus remains unchanged between times~$t-1$ and~$t$.

When the constant reference vector~$\ket{0}$ is also taken into account, Deutsch and Hayden's descriptors are complete. The expectation value~$\bra{0} \calO(t) \ket{0}$ of any observable~$\calO(t)$ that concerns only qubits of~$I$ can be determined by the descriptors~$\bol q_k(t)$, with~$k\in I$.
This can be seen more clearly at time~$0$, where an observable on the qubits of~$I$ is, like the gate~$G_t$, a linear 
operator that acts non-trivially \emph{only} on the qubits of~$I$. Again, any such operator can be generated additively and multiplicatively by the components of~$\bol q_k(0)$, with~$k\in I$. So $\calO(0) = f_{\calO}(\{\bol q_k(0)\}_{k\in I})$ and
\bes
\calO(t) = U^\dagger \calO(0) U = f_{\calO}(\{U^\dagger \bol q_k(0) U\}_{k\in I}) = f_{\calO}(\{\bol q_k(t)\}_{k\in I})\,.
\ees

\subsection{Noumenal States}

The abstract formalism of noumenal states~\cite{brassard2017equivalence} applies to more general theories than quantum theory. In this setting,
\emph{systems} are supposed to form a boolean algebra. Specifically, the union and the intersection of systems are systems, and there exist a \emph{whole system}~$S$ and an \emph{empty system}~$\emptyset$ with respect to which systems can be complemented, \emph{i.e.}, $\bar A$ satisfies $\bar A \cup A = S$ and $\bar A \cap A = \emptyset$. 

To each system~$A$ is associated a \emph{noumenal state}~$N^A$, a ``real state of affairs'', from which a \emph{phenomenal state}~$\rho^A$ can be determined by a surjective function, $\varphi(N^A) = \rho^A$. The phenomenal state encompasses all that can be observed, which may be informationally coarser than the noumenal state. In quantum theory, the phenomenal state boils down to the density matrix of the system, justifying the notation.

To system~$A$ is also associated a group of transformations~$\op A$ whose elements have an action on both noumenal\footnote{The action is faithful on noumenal states, which means that if $V \cdot N^A = \bar V \cdot N^A$ for all~$N^A$, then $V=\bar V$. } and phenomenal states. The function~$\varphi$ is promoted to a morphism, since it preserves the group action, namely, for any~$V \in \op A$,
\bes
\varphi (V \cdot N^A) = V \ast \rho^A \,,
\ees
where~$\cdot$ and ~$\ast$ denote the actions on noumenal and phenomenal states, respectively. The morphism~$\varphi$ also preserves the \emph{tracing out} of systems,
\bes
\varphi (\tr{B} N^{AB}) = \tr{B} \rho^{AB}\,,
\ees
where~$\tr{B}(\cdot)$ returns a state of system~$A$ from that of system~$AB$.
 The existence of a \emph{join product}, denoted~$\odot$, assures that any noumenal state of a composite system~$AB$ can be obtained by joining the local descriptions of~$A$ and of~$B$, 
\bes
N^{AB} = N^A \odot N^B \,.
\ees
This joining property, however, needs not to be preserved at the phenomenal level.
See figure~\ref{diagramme} for diagrammatic representations. 
\begin{figure}[h!]
	\centering
	\begin{tikzpicture} 
		\node (N) at (0,0){$N^A$};
		\node (UN) at (2,0) {$V \cdot N^A$};
		\draw[->] (N) -- (UN);
		\node at (0.9,0.3){$V$}; 
		
		\node (rho) at (0, -1.5){$\rho^A$};
		\node (Urho) at (2, -1.5){$V \ast \rho^A$};
		\draw[->] (rho) -- (Urho);
		\node at (0.9,-1.2){$V$}; 
		
		\draw[->] (N) -- (rho);
		\node at (-0.3,-0.8){$\varphi$};
		\draw[->] (UN) -- (Urho);
		\node at (1.7,-0.8){$\varphi$};
	\end{tikzpicture} \hspace{17pt}
	\begin{tikzpicture} 
		\node (N) at (0,0){$N^{AB}$};
		\node (UN) at (2,0) {$N^A$};
		\draw[->] (N) -- (UN);
		\node at (0.9,0.3){$\tr{B}$}; 
		
		\node (rho) at (0, -1.5){$\rho^{AB}$};
		\node (Urho) at (2, -1.5){$\rho^A$};
		\draw[->] (rho) -- (Urho);
		\node at (0.9,-1.2){$\tr{B}$}; 
		
		\draw[->] (N) -- (rho);
		\node at (-0.3,-0.8){$\varphi$};
		\draw[->] (UN) -- (Urho);
		\node at (1.7,-0.8){$\varphi$};
	\end{tikzpicture} \hspace{17pt}
	\begin{tikzpicture} 
		\node (NA) at (0,0){$N^{A}$};
		\node (NB) at (0.6,0.5){$N^{B}$};
		\node (NAB) at (2.5,0) {$N^{AB}$};
		\draw[->] (NA) -- (NAB);
		\draw[->] (NB) -- (NAB);
		\node at (1.11,0.18){$\odot$}; 
		
		\node (pA) at (0,-1.5){$\rho^{A}$};
		\node (pB) at (0.6,-1){$\rho^{B}$};
		\node (pAB) at (2.5,-1.5) {$\rho^{AB}$};
		\draw[->, dotted] (pA) -- (pAB);
		\draw[->, dotted] (pB) -- (pAB);
		\node at (1.11,-1.32){$\otimes$}; 
		
		\draw[->] (NA) -- (pA);
		\node at (-0.3,-0.8){$\varphi$};
		\coordinate (i1) at (0.6, 0.04);
		\coordinate (i2) at (0.6, -0.04);
		\draw (NB) -- (i1);
		\draw[->] (i2) -- (pB);
		\node at (0.3,-0.3){$\varphi$};
		\draw[->] (NAB) -- (pAB);
		\node at (2.2,-0.8){$\varphi$};
		
		\draw[very thick, red] (0.9,-1) -- (1.9,-1.7); 
		\draw[very thick, red] (0.9,-1.7) -- (1.9,-1);

%
	\end{tikzpicture}
	\caption{Evolution and tracing out are merely paralleled by noumenal and phenomenal states. However, noumenal states are endowed with a join product which permits to describe the whole from the parts, a feature that phenomenal states may not exhibit.}
	\label{diagramme}
\end{figure}
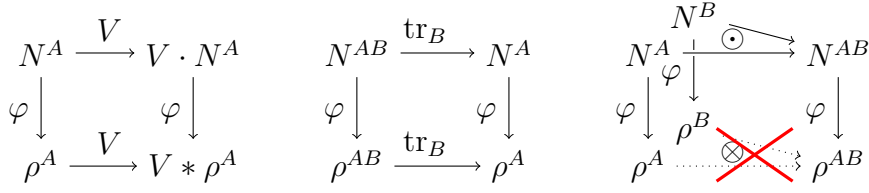

If~$V\in \op A$ and~$W\in \op B$, then the direct product $V\times W$, 
 defined by its action on local noumenal states as 
\bes
\left(V\times W\right ) \cdot N^{AB} = \left (V \cdot N^A\right ) \odot\left ( W \cdot N^B \right)\,,
\ees
is required to be a valid operation on~$AB$.
Transformations $U$ and $U'$ on the whole system~$S$ are~\emph{equivalent with respect to~$A$}, denoted $U \sim^A U'$, if they are connected by a transformation that acts trivially on~$A$,
\be \label{eqeqclass}
U \sim^A U' \iff \exists W \in \op{\bar A} \suchthat  U' = (\one^A \times W) U \,.
\ee
The noumenal state space associated to system~$A$ is defined as the set of equivalence classes with respect to the~$\sim^A$ relation.  
A particular noumenal state is then
\bes
N^A = [U]^A \,,
\ees
which encodes what has happened to the whole system~$S$ since the beginning, up to evolutions that do not causally concern system~$A$. From such a definition of the noumenal states, evolution by~$V \in \op{A}$, tracing out and joining are respectively defined as
\be \label{eqabstractdefs}
V \cdot [U]^A \deq [(V \times \one^{\comp A}) U]^A \,, ~~~
 \tr{B} [U]^{AB}\deq[U]^A ~~~ \text{and} ~~~
 [U]^A \odot [U]^B \deq [U]^{AB} \,. 
\ee
Finally, the morphism~$\varphi$ depends upon a reference phenomenal state~$\rho_0$ of system~$S$, and is defined as
\be \label{eqabsmor}
\varphi([U]^A) \deq \tr{\comp A}(U \ast \rho_0) \,.
\ee 
For the purposes of this paper, this abstract formalism needs to be instantiated by quantum theory. 
Let~$A$ and~$\bar A$ be complementary subsystems of the whole system~$S$, with their respective Hilbert spaces satisfying~$\hil A \otimes \hil{\bar A} \simeq \hil S$. 
The group of transformation~$\op A$ is then the group of unitary transformation acting on~$\hil A$, denoted~$\Uni (\hil{A})$. The quantum counterpart of the equivalence class \eqref{eqeqclass} allows to define what I shall call the~\emph{quantum noumenal state} as 
\be\label{eqqns}[U]^A= \left \{U' \in \Uni (\hil{S}) \suchthat   U' = (\one^A \otimes W) U \text{ for some } W \in \Uni (\hil{\comp A}) \right \} \,.
\ee
The locality of noumenal states follows directly from their defining equivalence classes: Any operation~$W$ performed on systems other than~$A$ leaves invariant the noumenal state~$[U]^A$. The completeness follows from the join product~$\odot$ and the morphism~$\varphi$.


\subsection{Evolution Matrices}\label{secqpl}

Evolution matrices~\cite{raymond2020local} are yet another way to achieve quantum locality. 
Let again~$A$ be a subsystem of~$S$, and let~$\{\ket{i}^A\}$ be a basis of~$\hil A$. As before, let~$U$ be the operation that occurred on~$S$ between time~$0$ and time~$t$. The evolution matrix\footnote{In contrast with 
ref.~\cite{raymond2020local}, I denote evolution matrices with double square brackets instead of single square brackets. This is to underline the distinction with the quantum noumenal state. Moreover, I reserve the appellation ``noumenal state'' for the objects defined by equivalence classes. 
}~$\llbracket U \rrbracket^A$ of system~$A$ is a matrix of dimension~$\dim (\hil A)$ whose elements are given by
 \bes
 \llbracket U \rrbracket^A_{ij} = U^\dagger (\ketbra{j}{i}^A \otimes \one^{\comp{A}}) U \,.
 \ees
 The dependence of the evolution matrix on~$U$ is only up to the~$\sim^A$ equivalence relation. Indeed, if~$U' =(\one^A \otimes W) U$,
\bea \label{eq:preuvepaul}
\llbracket U' \rrbracket_{ij}^A 
&=& U'^\dagger (\ketbra{j}{i}^A \otimes \one^{\comp{A}}) U' \nonumber \\
&=& U^\dagger (\one^A \otimes W^\dagger) (\ketbra{j}{i}^A \otimes \one^{\comp{A}}) (\one^A \otimes W)U \nonumber \\
&=& U^\dagger (\ketbra{j}{i}^A \otimes \one^{\comp{A}}) U \nonumber \\
&=& \llbracket U \rrbracket_{ij}^A \,, 
\eea
and one finds the same evolution matrix. The invariance of the evolution matrix within the equivalence class~$[\cdot]^A$  is necessary but insufficient to~\emph{identify} the evolution matrix with the quantum noumenal state. 
This identification shall be proved in the upcoming theorem~\ref{thmequiv2}. 
Still, this invariance assures the locality of evolution matrices.

Let~$A$ and~$B$ be disjoint systems. Evolution by~$V \in \Uni(\hil{A})$, tracing out and joining are defined as
\beas 
\left( V \llbracket U \rrbracket^A \right)_{ij} &\deq& \sum_{mn} V_{im} \llbracket U \rrbracket^A_{mn} V^\dagger_{nj}  \\
\left( \tr{B} \llbracket U \rrbracket^{AB} \right)_{ij} &\deq& \sum_k \llbracket U \rrbracket^{AB}_{ik;jk} \\
\left( \llbracket U \rrbracket^{A} \odot \llbracket U \rrbracket^{B} \right)_{ik;jl} &\deq& \llbracket U \rrbracket^{A}_{ij} \llbracket U \rrbracket^{B}_{kl}   \,.
\eeas
Phrased in the language of linear algebra, the above definitions differ from those of the abstract formalism, displayed in equations~\eqref{eqabstractdefs}, which instead find their evolution-matrix analogues in the following theorem.

\begin{theoA}[Raymond-Robichaud] \label{thmgilles}
Let~$A$ and~$B$ be disjoint systems and let~$V\in \emph{\Uni}(\hil{A})$. Then
\bes
V \llbracket U \rrbracket^A = \llbracket (V \otimes \one^{\comp A}) U \rrbracket^A \,,~~
 \tr{B} \llbracket U \rrbracket^{AB} \vphantom{ \left ( \llbracket^{AB} \right )}  = \llbracket U \rrbracket^{A} ~~~~\text{and} ~~~~
\llbracket U \rrbracket^{A} \odot \llbracket U \rrbracket^{B} = \llbracket U \rrbracket^{AB} \,.
\ees
\end{theoA}

The morphism~$\varphi$ is defined from a fixed reference density matrix~$\rho_0$ as
\be \label{eqdefrho}
\left ( \varphi \llbracket U \rrbracket^A \right )_{ij} \deq \tr{} \left( \llbracket U\rrbracket^A_{ij}\rho_0 \right)\,.
\ee
Again, this definition differs from its abstract counterpart, equation~\eqref{eqabsmor}, which finds its analogue in the following theorem.
Moreover, the theorem verifies that the morphism $\varphi$ intertwines evolution and tracing out, so that these relations are in fact paralleled by evolution matrices and density matrices. 
\begin{theoA}[Raymond-Robichaud] \label{thmepi}
Let~$A$ and~$B$ be disjoint systems and let~$V\in \emph{\Uni}(\hil{A})$. Then
\bes
\varphi \llbracket U \rrbracket^A = \tr{\comp A}(U\ast \rho_0) \,,~~
V \ast \varphi \llbracket U \rrbracket^A = \varphi (V\cdot \llbracket U \rrbracket^A) ~~\text{and}~~
\tr{B} \varphi \llbracket U \rrbracket^{AB} = \varphi \tr{B} \llbracket U \rrbracket^{AB} ,
\ees
where one recalls that in quantum theory, the action~$\ast$ is given by~\mbox{$U \ast \rho = U \rho U^\dagger$}.
\end{theoA}
\section{Equivalences} \label{secequiv}

In this section, I show that the three approaches to quantum locality presented in~\S\ref{secpre} are all formally equivalent. First, descriptors and evolution matrices are related by a mere change of operator basis. Second, the formalism of evolution matrices can be understood as a linear-algebraic model of quantum noumenal states, because, at least for qubits, the objects can be put in one-to-one correspondance. Establishing these equivalences is a milestone to later investigate the cost of these local descriptions.
\subsection{Descriptors $\leftrightarrow$ Evolution Matrices} \label{secdesem}

The equivalence between the formalisms of descriptors and evolution matrices is shown by considering an~$n$-qubit computational network~$\net$, since descriptors are defined in this setting.
As previously mentioned, the universality of such a quantum system trumps its apparent lack of generality. Let $\mathfrak Q_k$ denote the~$k$-th qubit.
%
At time~$t$, the descriptor of~$\mathfrak Q_k$ is given by
\bes
\bol q_k(t) 
= U^\dagger ( \one^{k-1} \otimes \sigma_x \otimes \one^{n-k} , \one^{k-1} \otimes \sigma_z \otimes \one^{n-k}) U \,,
\ees
while its evolution matrix is given by
\bes
\llbracket U \rrbracket^{\mathfrak{Q}_k}_{ij} = U^\dagger (\ketbra{j}{i} \otimes \one^{\comp{\mathfrak Q_k}}) U \,.
\ees
In both cases,~$U$ is the unitary operator according to which the network has so far evolved, and despite the different notation, the identity operators are applied on the same subspaces.
The objects $\bol q_k(t)$ and $\llbracket U \rrbracket^{\mathfrak{Q}_k}$ can easily be put in one-to-one correspondence, since they differ only by a change of operator basis; descriptors are expressed in the Pauli basis while evolution matrices are expanded in the canonical matrix basis. While the descriptor is composed of only two operators, $q_{kx}(t)$ and $q_{kz}(t)$, their multiplicative abilities permit the reconstruction of~$q_{ky}(t)= i q_{kx}(t) q_{kz}(t)$. Therefore,
\bes
\begin{array}{rclrcl}
&&&\qquad \llbracket U \rrbracket^{\mathfrak{Q}_k}_{11} &=& \dfrac{\one^{\otimes n} + q_{kz}(t)}{2} \\
q_{kx}(t) &=&  \llbracket U \rrbracket^{\mathfrak{Q}_k}_{12} + \llbracket U \rrbracket^{\mathfrak{Q}_k}_{21} 
&\llbracket U \rrbracket^{\mathfrak{Q}_k}_{12} &=& \dfrac{q_{kx}(t) - i q_{ky}(t)}{2} \\
q_{kz}(t)  &=& \llbracket U \rrbracket^{\mathfrak{Q}_k}_{11} - \llbracket U \rrbracket^{\mathfrak{Q}_k}_{22}
& \llbracket U \rrbracket^{\mathfrak{Q}_k}_{21} &=& \dfrac{q_{kx}(t) + i q_{ky}(t)}{2} \\
&&&\llbracket U \rrbracket^{\mathfrak{Q}_k}_{22} &=& \dfrac{\one^{\otimes n} - q_{kz}(t) }{2} \,.
\end{array}
\ees
%

The connection to observations is also equivalent in both formalisms. 
As argued in \S\ref{secpre}\ref{secdes} with the Heisenberg reference vector, the reference density matrix~$\rho_0$ can without loss of generality be fixed to~$\ketbra 00$. In fact, purity can be consecrated in the Church of the larger Hilbert space, \emph{i.e.}, if $\rho_0$ is not pure, it can be considered as a part of a larger network in a pure state. Furthermore, altering the global evolution~$U$ permits to fix the reference state to~$\ketbra 00$, as if some initial part of~$U$ was to prepare the network in the reference state~$\rho_0$.


The reduced density matrix~$\rho(t)= \tr{\comp{\mathfrak Q}_k}(U \ketbra 00 U^\dagger)$ of~qubit $\mathfrak Q_k$ at time~$t$ can be expressed in the Pauli basis as
$$
\rho(t) = \frac 12 \left (\one + \sum_{w \in \{x,y,z\}} p_w(t)  \sigma_w \right ) \,.
$$
From the trace relations of Pauli matrices, the components~$p_{w}(t)$ are
\bes
p_{w}(t) = \tr{}(\rho(t) \sigma_w) = \tr{}\left (U \ketbra 00 U^\dagger (\one^{k-1} \otimes \sigma_w \otimes \one^{n-k}) \right) = \bra 0 q_{kw}(t)\ket 0 \,.
\ees
The second equality from the left comes from that $\rho^A \mapsto \rho^A \otimes \one^B$ is, as a super-operator, the adjoint of $\rho^{AB} \mapsto \tr{B} (\rho^{AB})$, and the rightmost equality follows from cyclicality of the trace.

In the evolution-matrix framework, one can instead expand the reduced density matrix in its canonical representation
$
\rho(t) = \sum_{ij} \rho_{ij}(t) \ketbra ij \,.
$
The matrix elements can be obtained as
\bes
\rho_{ij}(t)= \tr{}(\rho(t)\ketbra ji) = \tr{}\left(U \ketbra 00 U^\dagger (\ketbra ji \otimes \one^{\comp{\mathfrak Q}_k}) \right) = \tr{} \left( \llbracket U\rrbracket^{\mathfrak Q_k}_{ij}\ketbra{0}{0} \right)\,,
\ees
consistently with the definition of the morphism, equation~\eqref{eqdefrho}.

\subsection{Quantum Noumenal States $\leftrightarrow$ Evolution Matrices}



The following theorem permits to identify, in the case of qubits, quantum noumenal states with evolution matrices. They are two~\emph{a priori} different ways of representing the same concept, the latter being a linear-algebraic model of the former.

\begin{theoA} \label{thmequiv2}
Let~$\net$ be an~$n$-qubit computational network, and let $\mathfrak Q_k$ denote the~$k$-th qubit. For all possible evolutions~$U$ and~$U'$ of~$\net$,
\bes
[U]^{\mathfrak Q_k} = [U']^{\mathfrak Q_k}~\iff \llbracket U\rrbracket^{\mathfrak Q_k} = \llbracket U'\rrbracket^{\mathfrak Q_k} \,.
\ees
\end{theoA}

\begin{proof}
The ``$\implies$'' was established by Raymond-Robichaud, and is presented in equation \eqref{eq:preuvepaul} of \S~\ref{secpre}.
Thanks to the equivalence between descriptors and evolution matrices, $\llbracket U \rrbracket^{\mathfrak Q_k} $ can be equivalently represented by
$$
\bol q_k(t) = U^\dagger \bol q_k(0) U \,.
$$
 To prove the ``$\Longleftarrow$'', assume $[U]^{\mathfrak Q_k} \neq [U']^{\mathfrak Q_k} $ and therefore,
 $U' \neq (\one^{\mathfrak Q_k} \otimes W) U$, for some~$W$ acting on~$\comp{\mathfrak Q_k}$. Hence, $U' = M U$, for some global operator $M$, whose functional representation~$\Uf_M(\bol q(0))$ depends explicitly on terms of $\bol q_k (0)$. But then, if $M$ is thought to occur between time $t$ and $t'$,
%
%
\beas
 \bol q_k (t') &=& U^\dagger M^\dagger \bol q_k(0) M U \\
&=& U^\dagger M^\dagger U U^\dagger \bol q_k(0) U U^\dagger M U \\
&=& \Uf^\dagger_M(\bol q(t)) \bol q_k(t) \Uf_M(\bol q(t)) \,.
\eeas
But because of its dependence on $\bol q_k(t)$, $\Uf_M(\bol q(t))$ acts nontrivially on~$\bol q_k(t)$ which changes it to a $\bol q_k(t') \neq \bol q_k(t)$, \emph{i.e.},
$
\llbracket U \rrbracket^{\mathfrak Q_k}  \neq  \llbracket U' \rrbracket^{\mathfrak Q_k} 
$.
\end{proof}

The previous theorem allows, at least for qubits\footnote{The analysis for qubits was eased by the formalism of descriptors, but the proof could be extended to systems of arbitrary dimensions. For such a system~$A$, the method is generalized by considering a generating set of 
operators acting on~$A$ and~$\comp A$. This can be achieved, for instance, with a generalization of Pauli matrices.
}, to unify Raymond-Robichaud's noumenal states defined by equivalence classes with his evolution matrices. The equivalence of \S \ref{secequiv}\ref{secdesem} puts also Deutsch and Hayden's descriptors into the same picture, a picture we can refer to as quantum locality.

\section{The Cost of Quantum Locality} \label{seccost}

A standard measure of complexity of an object that can continuously vary is the dimensionality of the space to which it belongs, or its number of degrees of freedom. After the density-matrix space of a qubit is recalled, the descriptor spaces of a single qubit and of the whole network are investigated.

\subsection{Density-Matrix Space of a Qubit}

Consider first the well-known example of the density-matrix space of a single qubit~$\mathfrak Q_k$ within an~$n$-qubit network~$\net$. The geometric object that characterizes such a state space, notwithstanding the size of the total system to which it belongs, is a unit ball in $\bbR^3$. This comes from the one-to-one correspondence between the density matrices of a qubit and the points on and inside the Bloch sphere, \textit{i.e.,}
$$
\rho = \frac 12 (\one + \vect p \cdot \vect{\sigma}) \,,
$$
where the polarisation 3-vector $\vect p$ 
is constrained by $|\vect p|\leq1$. 
 The space of density matrices of a qubit ranges along with the range of $\vect p$, which is the unit ball in $\bbR^3$, 
$$
\mathsf{Density}^{\mathfrak Q_k} \simeq D^3= \{\vect p \in \bbR^3 \colon |\vect p| \leq 1 \}\,.
$$
In particular,~$\dim (\mathsf{Density}^{\mathfrak Q_k}) = 3$.

\subsection{Descriptor Space of a Qubit in an~$n$-Qubit Network}

How big is the descriptor space --- or equivalently, the noumenal space --- of a qubit?
Unlike the density-matrix space, the dimension of the descriptor space of a qubit scales (exponentially!) with the size of the whole system~$\net$ to which it belongs. The proof of this assertion involves basic notions of Lie groups, which are smooth manifolds\footnote{In the present context, the reader who is unfamiliar with smooth manifolds can instead think of the special case of regular hypersurfaces. A hypersurface of dimension $n$ is an object defined by~$m$ independent constraints in~$\bbR^{n+m}$, \mbox{$\{y \in \mathbb{R}^{n+m} : F^a(y)=0\,,\;\; a= 1, \dots, m \}$}.
It is \emph{regular} if the~$m \times (n+m)$ matrix with elements $ \partial F^a / \partial y^i$ has full rank in all points. For more on manifolds, see for instance Ref.~\cite{lee2013smooth}.} endowed with a group structure.

From the equivalences established in \S\ref{secequiv}, the descriptor space of a qubit can be identified to the space of equivalence classes that define quantum noumenal states.
Define $H \subset \Uni(2^n)$ as the closed subgroup of operations of the form $\one^{\mathfrak Q_k} \otimes W$, where~\mbox{$W \in \Uni(2^{n-1})$} acts on~$\comp{\mathfrak Q_k}$. 
Therefore,
$$
\mathsf{Descriptor}^{\mathfrak Q_k \subseteq \net} \simeq \Uni \left(2^n \right)/ H 
 \,.
$$
For~$\bar k \neq k$, denote $\cnot_{\bar k \to k}$ the controlled-not gate in which qubit~$\bar k$ controls qubit~$k$ and denote~$N^{\bar k}$ the negation gate applied to~$\mathfrak Q_{\bar k}$. Then
$$
\cnot_{\bar k \to k} \left(\one^{\mathfrak Q_k} \otimes N^{\bar k}\right) \cnot_{\bar k \to k} \notin H \,,
$$
because it does not act trivially on $\mathfrak Q_{k}$, in particular, it changes $\ket{00}^{k\bar k}$ to $\ket{11}^{k \bar k}$. Because $\cnot$ is self-inverse, the above means that~$H$ is not a normal subgroup of~$\Uni(2^n)$, and so the quotient~$\Uni(2^n)/ H$ is not a group.
However, since $\Uni(2^n)$ is compact and its subgroup~$H$ is closed, the quotient manifold theorem~\cite[Thm 21.10]{lee2013smooth} is applicable and implies that the quotient $U(2^n)/H$ is a smooth manifold of dimension the difference of the dimensions of~$U(2^n)$ and~$H$.  
The group~$\Uni(N)$ has (real) dimension~$N^2$, because it is a hypersurface in $\bbC^{N^2} \simeq \bbR^{2N^2}$ subject to the~$N^2$ independent (real) constraints given by the real and imaginary parts of the system~$\sum_j u^*_{ji} u_{jk} = \delta_{ik}$, where~$i \le k$.  Since~$H \simeq \Uni(2^{n-1})$, one finds
\beas
	\dim\left(\mathsf{Descriptor}^{\mathfrak Q_k \subseteq \net}\right) 
	&=& \dim \left( \Uni(2^n) \right) - \dim \left( \Uni (2^{n-\tiny{1}}) \right)\\
	&=& 2^{2n} - 2^{2n-2}\\
	&=& \frac 34 \cdot 2^{2n} \,,
\eeas
and in particular, the dimension of the descriptor space scales exponentially with the size of the whole system~$\net$.

Compared to describing the three-dimensional reduced density matrix of a qubit, 
 if one instead faces the task of giving the~\emph{descriptor} of the same qubit, then one must feel like one has the Universe to describe.
This is in contradiction with the analysis by Hewitt-Horsman and Vedral~\cite[\S 3]{hewitt2007developing}, who claim that ``in general a given state defined by a density matrix has a unique representation in terms of Deutsch--Hayden operators''. This  statement hinges on a flaw in their analysis: In a nutshell, the number of constraints that determine a descriptor from a density matrix is overcounted, so the descriptor should be left underdetermined by the density matrix.

Notice that for such an~$n$-qubit network~$\net$ as a whole system, the universal wave function~$\ket \Psi$, \emph{i.e.}, the Schr\"odinger state of the whole network, has dimensionality $2^{n+1}-2$. Indeed, the amplitudes are fixed by $2 \cdot 2^n$ real parameters, but the normalization and the irrelevance of a global phase cut down two parameters. Therefore, \emph{the descriptor of a single qubit has larger dimensionality than the Schr\"odinger state of the whole network --- or of the Universe!}

The cost of describing a single qubit by its descriptor seems tremendously large. 
Compared with the neat description of a qubit by its reduced density matrix, where only three parameters are needed, the exponential scaling seems catastrophic.
However, this high cost of describing quantum states in a local and complete way was to be expected. 
 Indeed, a complete description should account for the necessary $2^{n+1}-2$ parameters of the universal wave function~$\ket\Psi$, and the most economical local repartition of those parameters into $n$ qubits must still leave $\sim 2^n/n$ parameters in each qubit! Thus the neat description of a qubit by its reduced density matrix might be practical to solely encode locally accessible information, but it is way too simplistic to carry with it the (noumenal) consequences of the complex web of previous interactions (entanglement) that fully describes a qubit.
\subsection{The Universal Descriptor}

If the descriptor of a single qubit has larger dimensionality than that of the universal wave function, then how big is the space of the descriptor of the whole system, \emph{i.e.}, the universal descriptor? 
 It turns out that it is not much bigger than the descriptor space of a single qubit. 
The previous analysis can be paralleled, with~$\net$ as the considered system, whose complement is the empty system~$\emptyset$. Hence, the subgroup~$H$ is the set of operations of the form~$\one^{\net} \otimes e^{i\phi}$, which can be identified to~$\Uni(1)$. Consequently,
\bes
\mathsf{Descriptor}^{\net} \simeq \Uni \left(2^n \right) / \Uni \left(1 \right) \disand \dim \left(\mathsf{Descriptor}^{\net} \right)= 2^{2n}-1 \,.
\ees
Therefore, the universal descriptor is, up to a phase, the unitary operator that occurred on the whole system from time~$0$ up to now. In this case, $H$ is a normal subgroup of~$\Uni(2^n)$, so~$\mathsf{Descriptor}^{\net}$ keeps a group structure. 

A more pedestrian approach can also be used to establish that the local descriptions of all qubits provide the knowledge of the evolution operator~$U$, up to a phase. 
Indeed, from the descriptors or evolution matrices of each qubit of the network, one can multiplicatively and additively reconstruct $U^\dagger \ketbra{j}{i} U$ for all~$i$ and $j$, where $\{\ket i\}_{i=0}^{2^n-1}$ is a basis of~$\hil{\net}$. The matrix element~$\ell,k$ of~$U^\dagger \ketbra{j}{i} U$ is given by
\bes
\bra \ell U^\dagger \ketbra{j}{i} U \ket k = u^*_{j\ell} u_{ik} \,.
\ees
 By setting~$i=j=k=\ell=0$, one finds~$|u_{00}|^2$, which can be assumed to be non-zero by otherwise permuting the columns of~$U$. By setting~$j=\ell=0$, but leaving~$i$ and~$k$ free, one finds~$u^*_{00} u_{ik}$ for all~$i$ and~$k$. Therefore, up to a phase ($u^*_{00}/|u_{00}|$), $U$ can be computed from $U^\dagger \ketbra{j}{i} U$ for all~$i$ and~$j$, which can be computed from $\bol q_i(t)$ or $\llbracket U \rrbracket^{\mathfrak Q_i}$ for all~$i$.
 
 If the initial state is again denoted~$\ket 0$, the universal wave function is given, up to a phase, by
\bes
\ket \Psi = U \ket 0 \,.
\ees 
In the right basis, it corresponds to a single column of the universal descriptor, which up to a phase is $U$, so $U\ket{0'}$ for all possible initial state~$\ket{0'}$.
If the multiplicity of classical-like terms in Everett's \emph{universal wave function} has prompt DeWitt and Graham\footnote{
In his work~\cite{everett1973theory, everett1957relative}, Everett never referred to ``many worlds''.}~\cite{dewitt2015many} to coin the~\emph{many worlds interpretation}, then the multiplicity of Everett's states in the universal descriptor can be thought as \emph{many many} worlds, \emph{i.e.}, as many ``many worlds'' as there are dimensions in the Hilbert space.

\section{More than the Universal Wave Function} \label{secwhatmore}

 The many-to-one correspondence between the universal descriptor and the global Schr\"odinger state has already been pointed out by Wallace and Timpson~\cite{wallace2007non}. They argued that since the descriptors corresponding to the same Schr\"odinger state lead to the same observations, they should be equated by some ``quantum gauge equivalence''. In such a case, the description left out boils down again to the usual Schr\"odinger state, retrieving non-locality.  In response, Deutsch~\cite{deutsch2011vindication} attacks the premise and argues that the \emph{dynamics} that has lead to such an actual Schr\"odinger state, too, may manifest in observations. In~\S 5 of his paper, he proposes a way in which one can tell apart different descriptors that yield the same Schr\"odinger state. His proposal sums up to gate by gate process tomography, which is consistent with our identification of the universal descriptor to the evolution operator. Still, it should be mentioned that the precise network contains yet more information than the evolution operator, since for instance the latter cannot distinguish the application of a gate followed by its inverse from the identity network. 
Raymond-Robichaud, also aware of the \emph{non-injectivity} of the morphism~$\varphi$ between noumenal and phenomenal states, 
 holds an intermediate standpoint that crops up in his nomenclature. The whole point of his approach is to oppose to the Wallace--Timpson identification and authorize --- in the name of locality --- the existence of noumenal states as elements of reality even if different such states may give rise to the same observations, 
  encompassed by the phenomenal state.

But what is the extra information contained in the universal descriptor~$\bol q(t)$ that is unobtainable from the universal wave function~$\ket \Psi$ alone? What does it operationally feature? As previously mentioned, this supplementary information can be thought to encode the universal wave function for any initial state vector. This is out of reach for the wave function alone, since $\ket \Psi = U \ket 0$ is of no use to determine $\ket{\Psi'} = U \ket{0'}$ for a different initial state, with $\braket{0}{0'} = 0$. However, $\bol q(t)$ can be used to compute this alternative universal Schr\"odinger state~$\ket{\Psi'}$, or, more in hand with the Heisenberg picture, the expectation~$\bra{\Psi'} \calO \ket{\Psi'} $ of any observable. A computation as such can be done by first defining a unitary operator~$V$ such that $V \ket 0 = \ket{0'}$, which could be thought of being performed between ``time $0$'' and ``time $0'$''. Recalling that~$\calO$ can be reconstructed from the components of $\bol q(0)$, it suffices to track the evolution of the latter.
\beas
\bra{\Psi'} \bol q(0) \ket{\Psi'} &=& \bra{0'} U^\dagger \bol q(0) U \ket{0'} \\
&=& \bra 0 V^\dagger \Uf_U^\dagger(\bol q(0)) V V^\dagger \bol q(0) V V^\dagger \Uf_U(\bol q(0)) V \ket{0} \\
&=& \bra 0 \Uf_U^\dagger(\bol q(0')) \, \bol q(0') \,\Uf_U (\bol q(0')) \ket 0 \,.
\eeas
Therefore, since $\bol q(t)=U^\dagger \bol q(0) U = \Uf_U^\dagger(\bol q(0)) \bol q(0) \Uf_U(\bol q(0))$ can be determined by a fixed function of $\bol q (0)$, then $V^\dagger U^\dagger \bol q(0) U V $ is determined by the same function, but instead evaluated on argument $\bol q(0')$. 
This puts in evidence a feature that the extra structure of the descriptor provides over the wave function, namely, it enables to evolve the description of the system in both directions of time, simultaneously. On the one hand, adding a gate at the end of the network affects the outer shell, that is to say, the function that determines $\bol q(t+1)$ from $\bol q(0)$ will differ from that of~$\bol q(t)$. On the other hand, supplementing a gate at the beginning of the network changes the inner shell: The defining function of~$\bol q (t)$ remains the same, but it is instead applied to the argument~$\bol q (0')$ --- itself a function of~$\bol q (0)$.

The dynamics corresponding to any possible initializations of the network, encompassed by the universal descriptor but not by the universal wave function, may also show useful to address an elephant in the room of unitary quantum theory.
In conventional quantum information theory, qubits initialized in a state~$\ket 0$ are taken as a free entity; but how does one really get such an initialized qubit in a unitary quantum realm, where irreversible processes such as $\ket ? \to \ket 0$ don't happen? How does one even unentangles the qubit from the rest of the Universe? This may be referred to as the~\emph{preparation problem}, dual to the measurement problem. A parsimonious solution provides a mechanism that explains, from within unitary quantum theory, why computations can be done the usual way \emph{as if} the initial state of the network really was~$\ket 0$. Such an explanation relies on decoherence arguments, and in the larger unitary scheme, not only~$\ket 0$ ``goes through'' the network, justifying the need for more dynamics to be tracked. Hence, computing in the Schr\"odinger picture with a fixed initial state~$\ket 0$ amounts to computing in a realm in which one neglects the possible interference effects that other parts of the multiverse may cause on ours. This is very well justified when such a rapidly decohered system (like an experimentalist) prepares the initial qubits but also in other parts of the multiverse such as in a quantum computer that happens to be performing classical operations~\cite{deutsch2002structure}.

The complexity 
 of the descriptor was investigated 
 through the dimensionality of its space, well motivated in physics. However, a computer theoretic approach may regard as the complexity cost of the descriptor its difficulty in time, in space or in program size to produce it. An investigation as such should be hand in hand with network complexity, since the whole descriptor is but a compact representation of its generating network. Perhaps, also, new insights into quantum Kolmogorov complexity could be provided by the local approach to quantum theory. 
 
If like Everett one is willing to take quantum theory for what it is and
accepts that the~$n$-qubit Universe is encoded in an object~$\ket \Psi$ moving in~$2^{n+1} - 2$ dimensions, then one should without regrets square this number up to~$2^{2n} - 1$ and instead use the universal descriptor for an entirely local 
story. One then faces the surprising consequence that more than~$3/4$ of the whole dimensionality resides in each qubit. Most of this information is locally inaccessible; it accounts for common histories among qubits, keeping track of who is entangled with whom. The consequence becomes more digestible when one appreciates how entangled the universe really is.
Bell stated~\cite{bell2004speakable} 
that ``Either the wave function, as given by the Schr\"odinger equation, is not everything, or it is not right''. It appears to be right in the limit of decohered prepared states, but not everything --- and completing it by the universal descriptor, which can be broken into spooky-action-free parts, is perhaps what Einstein, Podolsky and Rosen~\cite{einstein1935can} were looking for.

\vskip6pt

\enlargethispage{20pt}





\section*{Acknowledgements}
I am deeply grateful to Gilles Brassard for his well-meant guidance that valorized my research autonomy.
I am also grateful to Charles H. Bennett, Xavier Coiteux-Roy, David Deutsch, Pierre McKenzie, Jordan Payette
 and Paul Raymond-Robichaud for fruitful discussions and comments on earlier versions of this paper.
 I also wish to thank Stefan Wolf as well as the Institute for Quantum Optics and Quantum Information of Vienna, in particular Marcus Huber's group, for warm welcome and inspiring discussions.

This work was supported in part by the Fonds de recherche du Qu\'ebec -- Nature et technologie (FRQNT), the Swiss National Science Foundation (SNF), the National Centre for Competence in Research ``Quantum Science and Technology'' (NCCR \emph{QSIT}), the Natural Sciences and Engineering Research Council of Canada (NSERC) as well as Qu\'ebec's Institut transdisciplinaire d'information quantique (INTRIQ).

\newpage

\bibliographystyle{unsrt} 
\bibliography{refs} 

\end{document}